\documentclass[letterpaper,USenglish]{oasics-v2016}

\usepackage{microtype}

\title{Constrained Submodular Maximization via Greedy Local Search}
\titlerunning{Constrained Submodular Maximization}

\author[1]{Kanthi K. Sarpatwar}
\author[1]{Baruch Schieber}
\author[2]{Hadas Shachnai}
\affil[1]{IBM T. J. Watson Research Center, Yorktown Heights, New York, USA\\
  \texttt{\{sarpatwa, sbar\}@us.ibm.com}}
\affil[2]{Computer Science Department, Technion, Haifa, Israel\\
  \texttt{hadas@cs.technion.ac.il}}
\authorrunning{KK Sarpatwar, B Schieber, H Shachnai} 

\Copyright{Kanthi K. Sarpatwar, Baruch Schieber, Hadas Shachnai}


\usepackage{amsthm, amssymb, amsmath}
\usepackage{palatino}
\usepackage{setspace}
\usepackage{url}
\usepackage{graphicx}
\usepackage{float}
\usepackage{algorithm}
\usepackage{color}
\usepackage{epsfig}
\usepackage{amsmath}
\usepackage{amssymb}
\usepackage{algpseudocode}
\floatstyle{boxed}
\restylefloat{figure}

\def\eod{\vrule height 6pt width 5pt depth 0pt}

\newcommand{\negA}{\vspace{-0.05in}}
\newcommand{\negB}{\vspace{-0.1in}}

\newcommand{\mysubsection}[1]{\negB\subsection{#1}\negA}

\newcommand{\myparagraph}[1]{\par\smallskip\par\noindent{\bf{}#1:~}}

\newcommand{\be}{\begin{equation}}
\newcommand{\ee}{\end{equation}}

\def \R{I \!\! R}

\newcommand{\by}{{\bar y}}
\def \1{1 \!\! 1}

\newcommand{\cM}{{\cal M}}

\newcommand{\eps}{\varepsilon}

\newcommand{\amatk}{1-{\sc MatKnap}}
\newcommand{\amatik}{$k$-{\sc MatKnap}}

\newcommand{\addls}{{\tt {added}}}

\newcommand{\remove}[1]{}
\newcommand{\comment}[1] {}

\newtheorem{claim}[theorem]{Claim}

\newcommand{\cm}[1]{}

\newcommand{\cF}{{\cal F}}
\newcommand{\sopt}{{S^*}}
\newcommand{\bts}{{{\tilde B}^*}}
\newcommand{\stt}{{S^{t}}}
\newcommand{\stts}{{S^{t^*}}}
\newcommand{\rtto}{{\rho_{t+1}}}

\newcommand{\argmax}{\mbox{argmax}}

\def \ee   {\varepsilon}

\newlength{\tablength}
\newlength{\spacelength}
\newcommand{\tabstar}{\hspace*{\tablength}}
\newcommand{\spacestar}{\hspace*{\spacelength}}
\def\obeytabs{\catcode`\^^I=\active}
{\obeytabs\global\let^^I=\tabstar}
{\obeyspaces\global\let =\spacestar}
\newenvironment{display}{\begingroup\obeylines\obeyspaces\obeytabs}{\endgroup}
\newenvironment{prog}{\begin{display}\parskip0pt\sf}{\end{display}}

\def\obeytabs{\catcode`\^^I=\active}
{\obeytabs\global\let^^I=\tabstar}
{\obeyspaces\global\let =\spacestar}
\begin{document}

\maketitle

\begin{abstract}
We present a 
simple combinatorial $\frac{1 -e^{-2}}{2}$-approximation algorithm 
for maximizing a monotone submodular function subject to a knapsack and a matroid
constraint.
 This classic problem is known to be hard to approximate within factor better than $1 - 1/e$.
We show that the algorithm can be extended to yield a ratio of $\frac{1 - e^{-(k+1)}}{k+1}$ for the problem with a single knapsack and the intersection of $k$ matroid 
constraints, for any fixed  $k > 1$.

Our algorithms, which combine the greedy algorithm of  [Khuller, Moss and Naor, 1999] and [Sviridenko, 2004] with local search, 
show the power of this natural framework in submodular maximization with 
combined constraints.
\end{abstract}

\section{Introduction}
A set function $f\!\!:\! 2^U \rightarrow \R$ is submodular if for every $R,T \subseteq U$, $f(R) + f(T) \geq f(R \cup T) + f(R \cap T)$.
Such functions are ubiquitous in diverse fields,  
most notably in combinatorial optimization, operations research and  economics.
An equivalent definition of submodularity, which is perhaps more intuitive, refers to its diminishing returns: 
$f(R \cup \{ u \}) - f(R) \leq f(T \cup \{ u \} )- f(T)$, for any $T \subseteq R \subseteq U$, and $u \in U \setminus R$.
The concept of diminishing returns is widely used in economics, often leading to submodular utility functions.
Submodular functions also provide a unifying framework which captures optimization problems that have been studied earlier, such as Min Cut, Max Cut, Multiway Cut, Maximum Coverage, 
the Generalized Assignment Problem,  or the Separable Assignment Problem.

In these settings, the goal is to optimize a submodular function subject to certain constraints. In fact, in many cases the constraints at hand are quite simple, such as
knapsack constraints, matroid constraints, or a combination of the two.
Such submodular maximization problems naturally arise in advertising campaigns~\cite{AF+16}, combinatorial auctions~\cite{DS06,Vo08},
social networks~\cite{KKT03,KKT05},
and document/corpus summarization~\cite{SSSJ12,KB14}.

We consider first  the classic problem of maximizing a monotone submodular function subject to a knapsack and a matroid constraint. Formally, let $f$ be a non-negative monotone 
submodular function of a ground set $U$. Throughout the paper, we assume that $f$ is given via a value oracle; that is, given a set $S \subseteq U$, the oracle returns $f(S)$. 
Let ${\cF} \subseteq 2^U$ be a family of subsets of $U$, and ${\cM}= (U, {\cF})$ a matroid defined over $U$ and 
${\cF}$ (see Section \ref{sec:single_matroid} for the formal definition).

The subsets in the collection ${\cF}$ are called {\em independent sets}. 
Suppose that each element $u \in U$ has a nonnegative
size $c_u$, and let $B$ be a
given size bound.
Our goal is to maximize $f(S)$ over all
subsets $S\subseteq U$, such that $S$ is an independent set of ${\cM}$, and
the total size of the elements in $S$ is bounded by $B$, i.e.,
\begin{equation}
\label{eq:problem_def}
\max_{S \subseteq U} \{ f(S):~ S \in {\cal F}  \mbox{ and } \sum_{u \in S} c_u \leq B \}
\end{equation}

We further consider a generalization of (\ref{eq:problem_def}) in which we are given $k$ collections of independent sets, ${\cF}_1, \ldots, {\cF}_k$, each containing subsets
of $U$, and the corresponding matroids ${\cM}_j = (U, {\cF}_j)$, $j=1, \ldots , k$. The goal is to maximize $f(S)$ over all subsets $S$ of $U$ which are
independent sets in $\bigcap_{j=1}^k  {\cF}_j$. Formally, the optimization problem is 

\begin{equation}
\label{eq:k_matroid_problem}
\max_{S \subseteq U} \{ f(S):~ S \in \bigcap_{j=1}^k  {\cF}_j  \mbox{ and } \sum_{u \in S} c_u \leq B \}
\end{equation}

There is a beautiful line of research in this area. 
For a long time, the only known work 
had been a sequence of papers by Fisher, Nemhauser and Wolsey~\cite{NWF78,FNW78,nw78}, showing that a greedy algorithm 
achieves a ratio of $1 - 1/e$ to the optimum for maximizing a monotone submodular function
under a cardinality constraint,\footnote{The special case of a knapsack constraint where $c_u=1$ for all $u \in U$.} with a matching hardness of approximation result in the oracle model. The paper~\cite{FNW78} shows that simple local search 
yields a ratio of
$1/2$ when the function is maximized under a matroid constraint. A greedy algorithm was shown to achieve a ratio of $\frac{1}{k+1}$ for this problem with $k$ matroid constraints.

The hardness of approximation within ratio better than $1 -1/e$, already for Maximum Coverage
(i.e., maximizing a {\em linear} function under cardinality constraint),
follows from a result of Feige~\cite{f98}. This has led to an ongoing research aiming to develop approximation
algorithms which come close to this bound. The seminal paper of Khuller, Moss and Naor~\cite{kmn99} achieved the ratio of $1 - 1/e$ for Budgeted Maximum Coverage, using a greedy algorithm.
Their result inspired the later work of Sviridenko~\cite{s04}, presenting a simple greedy algorithm for maximizing a monotone submodular function subject to a knapsack constraint.
For a single matroid constraint, there are several algorithms achieving the ratio  $1-1/e$ via continuous greedy and multilinear relaxations~\cite{Vo08,ccpv10}
as well as other advanced techniques~\cite{FW12,BV14}.

For all of these cases the best known approximation ratios are essentially the best possible, with many of the algorithms easy to employ in practical scenarios; yet, already for maximizing a monotone submodular function subject to a knapsack {\em and}
 a matroid constraint, attempting to approach the ratio $1-1/e$ within additive of $\eps$, for any fixed $\eps > 0$, requires existing algorithms to perform $\Omega(n^{\mathrm{poly}(1/\eps)})$ steps~\cite{CVZ10} (see also~\cite{GNR16}).
 The dependence on $1/\eps$  in the exponent renders these algorithms impractical 
 even for $\eps=1/2$. The next best ratio of $0.38$, obtained via
 multilinear relaxations and contention resolution schemes~\cite{CVZ14},
deviates from the known hardness of approximation bound of $1 - 1/e$ almost by factor of 
$2$.\footnote{The bound of $1- 1/e$ follows from the known inapproximability results for each constraint alone~\cite{f98,FHK17}. }
A fast algorithm of~\cite{BV14} achieves the ratio $0.25$.
 Similarly, for a single knapsack constraint {\em and} $k$ matroid constraints, where $2 \leq k \leq 5$, the best known ratio is
$h(k,\eps) =\max_{0 \leq b \leq 1} (1 - e^{-b} - \epsilon)[(1 - e^{-b})/b]^k$, due to ~\cite{FNS11a} (see also~\cite{F13} and Table~\ref{table:compare_mult_matroid}). 
 For fixed $k >5$, the best ratio is $\frac{1}{k+3+\eps}$ 
 derived in~\cite{BV14}.\footnote{The algorithm of~\cite{BV14} achieves these bounds in almost {\em linear} time, for a class of more general instances.}
 The best known hardness of approximation result is  $O(\frac{\log k}{k})$, due to~\cite{HSS06}, for the intersection of $k > 1$ matroid constraints.
 
In this work, we amend the above state of affairs 
by obtaining practical algorithms whose approximation ratios are closer to the known hardness of approximation results.
We note that some of the 
previous algorithms have better running times, while the results of~\cite{CVZ10,GNR16}
lead to better bounds  (albeit at high computational cost). Yet, as detailed below, 
a notable advantage of our algorithms is in being extremely simple.

While we focus in this paper on maximizing a monotone submodular function under a single knapsack and (single or multiple) matroid constraints, there are known results for extension to 
intersection of $k$-system and $\ell$ knapsack constraints~(e.g., \cite{GNR10,GR+10a,GRST10,BV14,CVZ14}), and for maximizing a 
non-monotone submodular 
function~\cite{FMV07,LSV10,LMNS10a,GRST10,FNS11,FNS11a,KST13,CVZ14,FHK17}. 
For unconstrained submodular maximization, the best known result is
a randomized linear time $1/2$-approximation algorithm by
Buchbinder et al.~\cite{BFNS15}.

\myparagraph{Our Contribution}
We present combinatorial algorithms for maximizing a monotone submodular function subject to a knapsack and $k$ matroid constraints, yielding
approximation ratio of $\frac{1-e^{-(k+1)}}{k+1}$, for any fixed $k \geq 1$. 
The algorithms are based on the greedy approach of~\cite{kmn99,s04}
combined with local improvement steps which allow $k$-swaps (see, e.g., \cite{LMNS10a}).  These easy
features make the algorithms 
highly intuitive in tackling our problems.
For the single matroid case, we show that 
our algorithm can be implemented in ${\tilde{O}}(n^6)$ time, where $n=|U|$ is the cardinality of the ground set.\footnote{The notation $\tilde{O}$ ``hides'' poly-logarithmic factors.} 
We note that the greedy algorithm of Sviridenko~\cite{s04}, that yields the
best possible ratio of $1-1/e$ for maximizing a submodular function subject to a knapsack constraint, 
has running time $O(n^5)$.
For a single knapsack and $k$ matroid constraints, where $k >1$ is fixed, the running time of our algorithm is 
$n^{O(k)}$. Table~\ref{table:known_results_single_matroid} gives the known results for the special case where $k=1$, along
with the running times of the algorithms.
 Table~\ref{table:compare_mult_matroid} compares our approximation ratios with best previous results 
for several values of $k >1$.

Our algorithms handle the combined knapsack and matroid constraints by applying {\em greedy swaps}. While the greedy property
guarantees that we do not exceed the knapsack constraint without collecting enough value,  the swaps maintain the independence of the solution set 
given the matroid constraint. 
Our algorithms for single and multiple matroid constraints (outlined in Alg.~\ref{alg:submod_km} and Alg.~\ref{alg:submod_kmi}, respectively) proceed in the same fashion. 
The only difference is that in the multiple matroid case, we replace the `swaps' with `$k$-swaps'; a $k$-swap
may involve the deletion of up to $k$ elements from the 
solution set, while adding to the set a single element.
To the best of our knowledge, this notion of {\em greedy swaps} is used here for the first time.


\begin{table}[h!]
\centering
\parbox{0.4\textwidth}{
\begin{footnotesize}
 \begin{tabular}{|l |l|} 
 \hline
 {\bf Approximation ratio}&{\bf Running time} \\ [1ex] 
 \hline
$1-1/e -\eps$ & $\Omega(n^{\frac{1}{\eps^4}})$~\cite{CVZ10} \\
 \hline
$\frac{1 - e^{-2}}{2 } \approx {\bf 0.432}$&  ${\tilde{O}}(n^6)$~{\bf This paper}\\ 
 \hline
$0.38$ & $poly(n)$~\cite{CVZ14}\\ 
\hline
$0.25$& $O(\frac{n}{\eps^2} \log^2 \frac{n}{\eps})$~\cite{BV14}\\ 
\hline
\end{tabular}
\end{footnotesize}
\centering
}
\caption{Known results for maximizing a monotone submodular function subject to a knapsack and a matroid constraint.
The ratio in~\cite{CVZ10} becomes strictly better than the ratio in this paper by taking $\eps < 0.2$.
}
\label{table:known_results_single_matroid}
\end{table}

\begin{table}[h!]
\centering
\parbox{0.4\textwidth}{
\begin{footnotesize}
 \begin{tabular}{|l |l |l|} 
 \hline
 ${\mathbf k}$ & {\bf Our Work}&{\bf Previous Best
 } \\ [0.5ex] 
 \hline
2 &$\frac{1 - e^{-3}}{3 } \approx $ {\bf 0.316}& $h(2,\eps) < 0.262$~\cite{FNS11a}\\ 
 \hline
3& $\frac{1 - e^{-4}}{4 } \approx ${\bf 0.245}& $h(3,\eps) < 0.192$~\cite{FNS11a}\\ 
 \hline
 5 &$\frac{1 - e^{-6}}{6 } \approx $ {\bf 0.166}& $h(5,\eps) < 0.127$~\cite{FNS11a}\\ 
 \hline
8 &$\frac{1 - e^{-9}}{9 } \approx $ {\bf 0.111}& $\frac{1}{11+\eps} < 0.091$ ~\cite{BV14}\\ 
\hline
9 &$\frac{1 - e^{-10}}{10 } \approx $ {\bf 0.1}& $\frac{1}{12+\eps} < 0.083$ ~\cite{BV14}\\ [1ex] 
 \hline
\end{tabular}
\end{footnotesize}
\centering

}
\qquad
\begin{minipage}[c]{0.53\textwidth}%
\centering
    \includegraphics[width=1\textwidth]{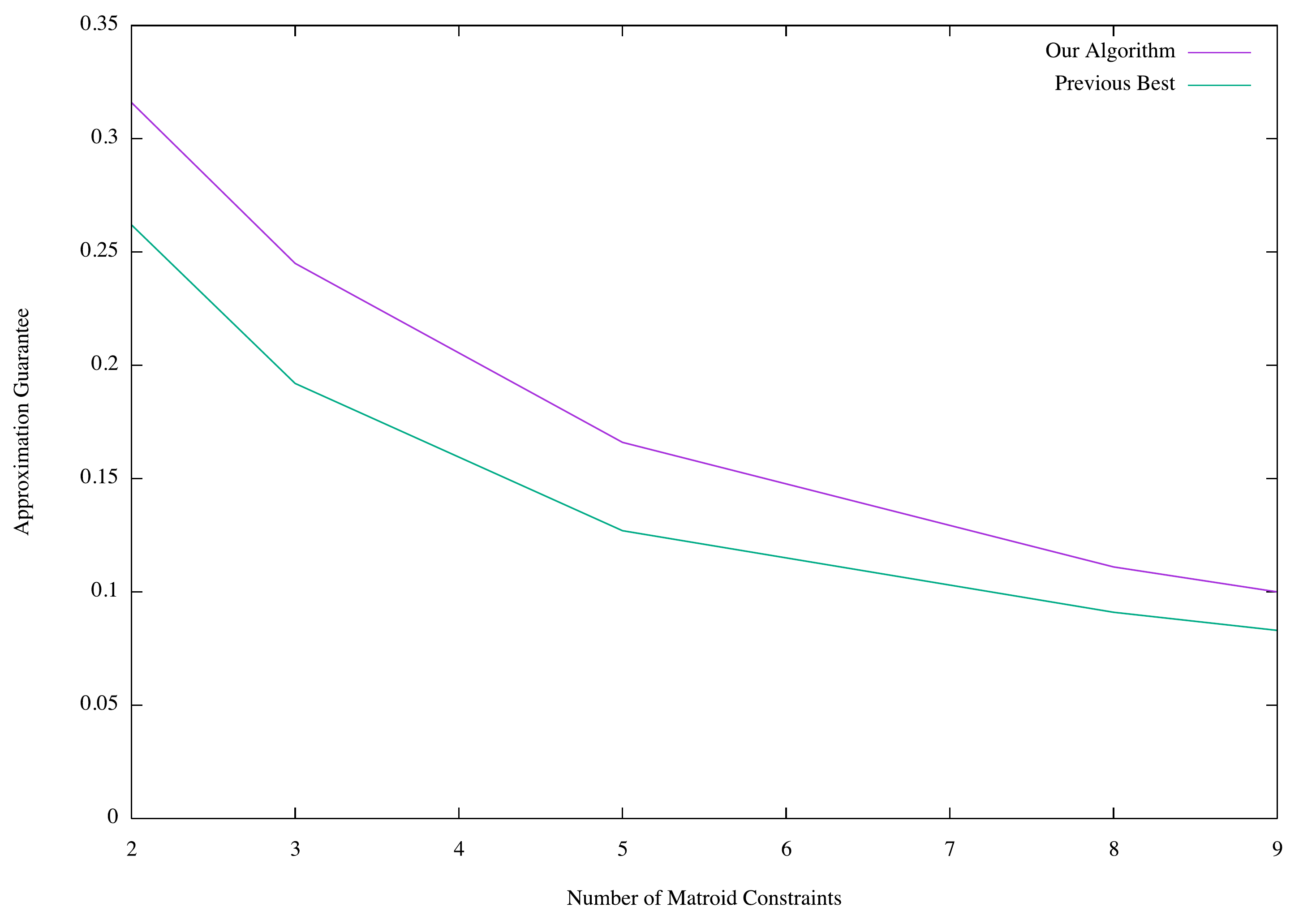}

\label{fig:compare_bounds}

\end{minipage}
\caption{Comparison of ``our work v\/s previous best'' for monotone submodular maximization subject to a single knapsack and $k>1$ matroid constraints.
The entries in the Previous Best column apply to the problem with $O(1)$ knapsack constraints.}
\label{table:compare_mult_matroid}
\end{table}

\section{A Single Matroid Constraint}
\label{sec:single_matroid}
We first study the problem where we have a knapsack constraint and a single matroid constraint (as defined in Eq.~(\ref{eq:problem_def})).

We start with some definitions and notation. For a subset of elements $U' \subseteq U$, let $c(U')$ denote the total size of the elements in $U'$, i.e.,
$c(U')= \sum_{u \in U'} c_u$.

For notational convenience, we consider $\phi$ as a dummy element not in the ground set $U$. We assume that $c_\phi=0$, and for any set $U' \subseteq U$,
$f(U' \cup \{ \phi \}) = f(U')$.
Given a subset $U' \subseteq U$, where $U'$ is independent, i.e., $U'\in {\cF}$,
a pair of elements  $(x,y)$, where  $x \in U \setminus U'$ and  $y \in U' \cup \{ \phi \}$,  is a {\em swap} if $(U' \setminus \{ y \}) \cup \{ x \}$ is independent, i.e.,
$(U' \setminus \{ y \}) \cup \{ x \} \in {\cF}$.
Let $L(U')$ denote the collection of
 swaps that involve elements from $U'$.
The marginal profit density of a swap $(x,y)$ is given by
$\rho_{(x,y)} = \frac{f((U'\setminus\{ y \}) \cup \{ x \}) - f(U')}{c_x}$.

\subsection{The Algorithm}
\label{sec:alg_mat_k}

Let $\sopt = \{u_1,u_2,\ldots, u_p\}$ be an optimal solution of the
given instance, ordered such that $u_{i+1}$ is the element with maximum marginal profit with
respect to the prefix $\{u_1,u_2,\ldots, u_{i}\}$.
We assume that $|\sopt| \geq 2$, since in case $\sopt$ consists of a single element, this element can be guessed to yield an optimal solution.
The algorithm starts by guessing $Y= \{ u_1, u_2 \}$, the two elements
with largest marginal profits in $\sopt$.

We set the initial solution to consist of  the set $Y$. The algorithm then  applies iterations of `greedy swaps'
that expand the solution while maintaining the knapsack and the matroid constraint.

We note that for ease of understanding the pseudocode of {\amatk} (in Algorithm~\ref{alg:submod_km}) does not
guarantee that the algorithm terminates in polynomial time. However, in Section~\ref{sec:runtime_single_matroid} we discuss a modified version of our algorithm that runs in polynomial time with
no harm to the approximation ratio.
\begin{algorithm}[!h]
\caption{{\amatk}$(U, B, {\cM})$ }
    \label{alg:submod_km}
  \begin{algorithmic}[1]
    \State Guess $u_1 = \argmax_{u\in\sopt} f(\{u\})$ and $u_2 = \argmax_{u\in\sopt\setminus{\{u_1\}}} f(\{u, u_1\}) - f(\{u_1\})$
    \State Let $Y = \{ u_1, u_2 \}$
    \State Initialize $S= Y$ and ${\addls} = true$
    \While{${\addls}$}
    \label{alg1:start_iteration}
    \Statex \Comment{greedy swaps}
        \State $\addls = false$
        \State Generate the collection of swaps $L=L(S)$
       \While{($not({\addls})$ and $L \neq \emptyset$)} \label{step:greedy-swap}
     \State Pop (i.e., pick and remove) from $L$ a swap $(x,y)$ with a maximum value of $\rho_{(x,y)}$
                \If{ ($y\notin Y$ and $\rho_{(x,y)} >0$ and  $c_x  - c_y + c(S) \leq B$)}
                \State $S= (S  \setminus \{y \}) \cup \{x\}$
                 \State $c(S)= c(S)  - c_y + c_x$
                 \State  ${\addls}= true$
         \EndIf
     \EndWhile
    \EndWhile
\State      \Return S
       \end{algorithmic}
\end{algorithm}

\subsection{Analysis}
\label{sec:analysis_single_matroid}
We now show that Algorithm {\amatk} yields an approximation ratio that is close to the optimal of $1-e^{-1}$ known for submodular maximization with
either a knapsack {\em or} a matroid constraint. Formally, our main result is the following.
\begin{theorem}
\label{thm:amatk_appx_ratio}
{\amatk} is a $\frac{1-e^{-2}}{2}$-approximation algorithm for
submodular maximization subject to a knapsack and a matroid constraint.
\end{theorem}
The following simple lemmas will be useful in the proof of Theorem~\ref{thm:amatk_appx_ratio}.
Let ${\cF}$ be a collection of subsets of $U$, and
${\cM}= (U,{\cF})$ a matroid defined on $U$. Recall that ${\cM}$ satisfies the following (see, e.g.,~\cite{CLRS01}).
\begin{enumerate}
\item[$(i)$]
{\em Non-emptiness:} The empty set is in ${\cF}$ (thus, ${\cF}$ is not itself empty).
\item[$(ii)$]
{\em Hereditary} property: If a set $R$ is in ${\cF}$ then every subset of
$R$ is also in ${\cF}$.
\item[$(iii)$]
{\em Exchange} property: If $R$ and $T$ are two sets in ${\cF}$, where $|R| > |T|$, then there is an element
$r \in R \setminus T$ such that $T \cup \{ r \}$ is in ${\cF}$.
\end{enumerate}
We note that all maximal independent sets (or {\em bases}) of a matroid have the same cardinality. The first lemma follows directly from the  above Exchange property.
\begin{lemma}
\label{lemma:matroid_bases}
Let $B_1,B_2$ be two bases of a  matroid ${\cM}=(U, {\cF})$, then for any
$x \in B_2 \setminus B_1$, there exists $y \in B_1 \setminus B_2$ such that
$(B_1 \setminus \{y\}) \cup \{x\}$ is a base of ${\cM}$.
\end{lemma}

\begin{lemma}
\label{lemma: mapping_sopt_s}
Given two independent sets $S,T \in {\cF}$, there exists a mapping
$b: T\setminus S \rightarrow (S\setminus T)\cup\{\phi\}$,
such that $(S \setminus \{b(u )\}) \cup \{ u \} \in {\cF}$,
for all $u \in T\setminus S$,
and $|b^{-1}(y)| \le 1$, for all $y \in S\setminus T$.
\end{lemma}
\begin{proof}
Let ${\bts}$ be a base of ${\cM}$, such that $T \subseteq {\bts}$.
By the Exchange property, $|{\bts}|-|S|$ elements from
${\bts} \setminus S$ can be added to the set $S$ to obtain a base of ${\cM}$.
Denote this base by  ${\tilde B}$.
For each element $x \in T\cap({\tilde B}\setminus S)$,
define $b(x)=\phi$.

We now iterate over all elements in
$T\setminus  {\tilde B} \subseteq {\bts} \setminus  {\tilde B}$.
Let $x\in T\setminus  {\tilde B}$ be the current element.
By Lemma \ref{lemma:matroid_bases},
there exists $y \in {\tilde B} \setminus {\bts}$, such that
$({\tilde B}\setminus \{y\}) \cup \{x\}$ is a base of ${\cM}$.
Note that $y$ must be in $S\setminus T$, since
${\tilde B} \setminus (S\setminus T) \subseteq {\bts}$.
We set ${\tilde B} = ({\tilde B}\setminus \{y\}) \cup \{x\}$ and $b(x)=y$,
and iterate on the next element in $T\setminus  {\tilde B}$.
By our construction, no two elements in
$T \setminus {\tilde B}$ are mapped to the same element in $S\setminus T$.
\end{proof}

The next lemma is due to Wolsey~\cite{w82}.
\begin{lemma}
\label{lemma:wolsey_ineq}
Given two positive integers, $P$ and $D$, and a set of  real-valued non-negative numbers $\gamma_i$, $i=1, \ldots , P$,
$$
\frac{\sum_{i=1}^P \gamma_i}{\min_{t\in \left[1..P\right]} (\sum_{i=1}^{t-1} \gamma_i + D \gamma_t)} \geq 1-\left(1- \frac 1D\right)^P \geq  1- e^{-P/D}
$$
\end{lemma}
The following generalizes a result of Sviridenko~\cite{s04}.
We include the proof here for completeness.
\begin{lemma}
\label{lemma:f_on_Y_swaps}
Given an element $u_\ell \in {\sopt}$,  $\ell \geq 3$, and
a subset $W \subseteq U \setminus \{ u_1, u_2, u_\ell\}$, for any swap  $(u_\ell ,w)$, such that  $w \in W \cup \{  \phi \}$,
$
f((Y \cup W  \setminus \{ w \}) \cup \{ u_\ell \} ) - f(Y \cup W ) \leq f(Y)/2.
$
\end{lemma}
\begin{proof}
    The ordering of the elements of ${\sopt}$, and the fact that $f(\cdot)$
    is submodular, monotone, and non-negative imply that for
$u_\ell \in {\sopt}$, where $\ell \geq 3$ and the subsets $Y$ and $W$, the following inequalities are satisfied:

\begin{align*}
    f\left((Y \cup W \setminus \{ w \}) \cup \{ u_\ell \}\right) - f(Y \cup W) & \leq
    f\left((Y \cup W \setminus \{ w \}) \cup \{ u_\ell \}\right) - f(Y \cup W \setminus \{w \})\\
& \leq  f(\{ u_\ell \}) - f(\emptyset)
\leq f(\{ u_1 \}),
\end{align*}
and
\begin{align*}
    f\left((Y \cup W \setminus \{ w \}) \cup \{ u_\ell \}\right) - f(Y \cup W) & \leq
    f\left((Y \cup W \setminus \{ w \}) \cup \{ u_\ell \}\right) - f(Y \cup W \setminus \{w \})\\
    & \leq f(\{ u_1 \} \cup \{ u_\ell \}) - f(\{ u_1 \}) \\
    & \leq f(\{ u_1, u_2 \} = Y) - f(\{ u_1 \}),
\end{align*}

Summing the two inequalities, we have
that
\begin{align*}
    2\left(f\left((Y \cup W \setminus \{ w \}) \cup \{ u_\ell \}\right) -f(Y \cup W)\right) & \leq
    f(Y) -f(\{ u_1 \}) + f(\{ u_1 \}) 
   =f(Y)
\end{align*}
\end{proof}

\begin{proof}[Proof of Theorem \ref{thm:amatk_appx_ratio}]
Suppose we have guessed
$u_1$ and $u_2$ correctly.
Let $S\subseteq U$ denote the subset output by the algorithm.
\remove{
WLOG assume that $\sopt$ is a basis of ${\cM}$.
If this is not the case we can add elements to  $\sopt$ to make it a basis
and this would not decrease the value of $f(\cdot)$ due to its monotonicity.
If $S$ is not a basis, let $S'$ be the extension of $S$ to a basis.
By the Bijective exchange Axiom~\cite{B69} (see also~\cite{B73}), there exists a bijection
$b': {\sopt} \rightarrow S'$, such that $(S' \setminus \{b(u )\}) \cup \{ u \} \in {\cF}$, for all $u \in {\sopt}$.
}
By Lemma~\ref{lemma: mapping_sopt_s}, there exists
a mapping $b: {\sopt}\setminus S \rightarrow (S\setminus\sopt)\cup\{\phi\}$,
such that $(S \setminus \{b(u )\}) \cup \{ u \} \in {\cF}$, for all $u \in {\sopt}\setminus S$,
and $|b^{-1}(y)| \le 1$, for all $y \in S\setminus{\sopt}$.

Throughout the analysis, we refer to {\em iterations} of the outer loop (Line~\ref{alg1:start_iteration}.) in Algorithm~\ref{alg:submod_km}.
Thus, each iteration, except maybe the last one, consists of a single greedy swap.
We distinguish between two cases, based on the last iteration of the algorithm.

\par\noindent{\bf Case 1:}
For all elements $x\in \sopt\setminus S$,
the swap $(x,b(x))$ was rejected in the last iteration because it satisfied $\rho_{(x,b(x))} \le 0$,
and not because $c_x  - c_{b(x)} + c(S) > B$.
That is, the swap $(x,b(x))$ did not violate the knapsack constraint.

From Case 1 assumption, it follows that no swap $(x,b(x))$, for
an element $x \in \sopt\setminus S$, could increase the value of $f(\cdot)$.
Hence, we have that
\begin{align}
\label{eq:LS}
f(\sopt)-f(S) & \leq \sum_{x \in\sopt\setminus S} \left( f(S \cup \{ x \}) -f(S)\right) \nonumber \\
            & \leq \sum_{x \in\sopt\setminus S}  \left( f(S \cup \{x \} \setminus  \{b(x) \}) - f(S \setminus \{ b(x) \})\right) \nonumber \\
            & \leq \sum_{x \in\sopt\setminus S}  \left( f(S) - f(S \setminus \{ b (x) \}) \right) .
\end{align}

The first and second inequalities follow from submodularity,
and the third from our assumption about the swaps.
Recall that for every $y\in S\setminus\sopt$ we have $|b^{-1}(y)|\le 1$, thus the (multi) set
$\{b(x)\}_{x\in\sopt\setminus S}$ does not contain any element from $S\setminus\sopt$ with
multiplicity higher than 1.
Let $\{b(x)\}_{x\in\sopt\setminus S}\setminus \phi=\{u_1, \ldots , u_K\}$,
where $K\le |\sopt\setminus S|$ and the indices are assigned arbitrarily.
Then, by  (\ref{eq:LS}), and using submodularity and monotonicity, we have
\begin{align}
\label{eq:LS_single_matroid}
f(\sopt)-f(S) & \le \sum_{x \in\sopt\setminus S}  \left( f(S) - f(S \setminus \{ b (x) \})\right)  \nonumber \\
& \leq
\sum_{i=1}^{K} \left( f(\{u_1, \ldots , u_i\}) - f(\{u_1, \ldots , u_{i-1}\})\right) \nonumber \\
&\le f(\{u_1, \ldots , u_{K}\}) \le f(S).
\end{align}
Hence, in this case, we have that {\amatk} yields a $\frac 12$-approximation to the optimum.

\par\noindent{\bf Case 2:}
At least one swap $(x,b(x))$, for $x\in \sopt\setminus S$, considered in
the last iteration,
violated the knapsack constraint; namely,
$c_x  - c_{b(x)} + c(S) > B$.

Denote by $S^\ell$ the subset of elements in the solution after iteration $\ell$
of the algorithm, for $\ell \geq 0$, where $S^0=Y$.
Let $b_\ell: \sopt\setminus S^\ell \rightarrow (S^\ell\setminus\sopt)\cup\{\phi\}$
be the mapping guaranteed by
Lemma~\ref{lemma: mapping_sopt_s}.
We note that in all iterations except the last the algorithm performs a swap
$(x_t,y_t)$  (where $y_t$ may be $\phi$).
Denote $\rho_{(x_t,y_t)}$ by $\rho_t$.

Let $t^*+1$ be the first iteration in which
a swap $(u_{t^*+1},b_{t^*}(u_{t^*+1}))$, for $u_{t^*+1}\in \sopt\setminus\stt$,
was considered in iteration $t^*+1$ and was rejected since it
violated the knapsack constraint. Clearly, such an iteration exists by the assumption
of Case 2.
By the choice of $t^*$, for all $t=0, \ldots , t^*-1$,
for every $x\in \sopt\setminus\stt$ such that $(x,b_t(x))$ is a swap, we must have
$\rho_{(x,b_t(x))} \le \rho_{t+1}$.

Define the function $g(S)=f(S)-f(Y)$. Since $f(\cdot)$ is submodular and monotone,
so is $g(\cdot)$.
Hence,  for all $t=0, \ldots , t^*$,
$$
g({\sopt}) \leq  g({\stt}) + \sum_{u \in {{\sopt} \setminus {\stt}}} ( g({\stt} \cup \{ u \}) -g({\stt})).
$$

Using submodularity, we can now write:
\begin{align}
\label{eq:partition}
g({\sopt}) - g({\stt})
& \leq \displaystyle {\sum_{x \in {{\sopt} \setminus {\stt}}}}
\left(g({\stt} \cup \{ x \}) - g({\stt})\right)  \nonumber\\
& \leq \displaystyle {\sum_{x \in {{\sopt} \setminus {\stt}}}}
\left(g(({\stt} \setminus \{ b_t(x) \})  \cup \{ x \}) - g({\stt} \setminus \{ b_t(x) \}) \right)  \nonumber\\
& \leq \displaystyle {\sum_{x \in {{\sopt} \setminus {\stt}}}}
\left(g({\stt}) -  g({\stt} \setminus \{ b_t(x) \}) \right)+
\left(g(({\stt} \setminus \{ b_t(x) \}) \cup \{ x \})  - g({\stt}) \right)
\end{align}
\begin{claim}
\label{claim:bound_R1_plus_R2}
$\displaystyle {\sum_{x \in {{\sopt} \setminus {\stt}}}}
(g({\stt}) - g({\stt} \setminus \{ b_t(x) \})) \leq g({\stt})$.
\end{claim}
\begin{proof}
Clearly, for any $x \in {{\sopt} \setminus {\stt}}$ such that $b_t(x)=\phi$,
we have $g({\stt}) - g({\stt} \setminus \{ b_t(x) \})=0$. Thus, we need only to consider the sum over the subset $C \subseteq {{\sopt} \setminus {\stt}}$, such that for any $x\in C$,  $b_t(x)\neq\phi$.
As noted above, the (multi) set $\{b_t(x)\}_{x\in C}$
does not contain any element from $\stt\setminus\sopt$ with
multiplicity higher than 1.
Let $\{b_t(x)\}_{x\in C}=\{u_1, \ldots , u_K\}$,
where $K= |C|$, and the indices are assigned arbitrarily.
Then, similar to Inequality~(\ref{eq:LS_single_matroid}),
\begin{align*}
 \displaystyle {\sum_{x \in C}}  (g({\stt}) - g({\stt} \setminus \{ b_t(x) \}))
 & \leq \displaystyle { \sum_{i=1}^K}(g(u_1, \ldots , u_i) -g(u_1, \ldots , u_{i-1}) )
 \leq g({\stt})
\end{align*}
\end{proof}
Using Claim~\ref{claim:bound_R1_plus_R2}, we have for any $t=0, \ldots , t^*-1$,

\begin{align}
\label{eq:rho}
g({\sopt}) & \leq 2g({\stt})
+ \sum_{x \in {{\sopt} \setminus {\stt}}}
\left( g({\stt} \setminus \{ b_t(x) \} \cup \{ x \})  - g({\stt}) \right) \nonumber\\
&= 2g({\stt})
+ \sum_{x \in {{\sopt} \setminus {\stt}}}
\left( f({\stt} \setminus \{ b_t(x) \} \cup \{ x \})  - f({\stt}) \right)  \nonumber\\
& \leq  2g({\stt}) + (B - c(Y)) \rho_{t+1}.
\end{align}

The last inequality is due to the following:
\begin{enumerate}
\item[(a)]
    By our choice of $t^*$,
    for any swap $(x,b_t(x))$,
    $f({\stt} \setminus \{ b_t(x) \} \cup \{ x \}) - f({\stt}) \leq c_x  {\rtto}$.
\item[(b)] $\sum_{x \in {\sopt} \setminus {\stt}} c_x \leq B - c(Y)$.
\end{enumerate}
Now, from the above discussion, we have for any $t=0, \ldots , t^*-1$,
\begin{equation}
\label{eq:sopt_vs_stt_greedy}
g({\sopt}) \leq 2  \left ( g({\stt}) +\frac{(B - c(Y))}{2} {\rtto} \right  ).
\end{equation}
For $t=1, \ldots , t^*$,
let $B_{t} = \sum_{\tau=1}^t (g(S^\tau)-g(S^{\tau-1}))/\rho_\tau$, and $B_0=0$.
Note that $(g(S^\tau)-g(S^{\tau-1}))/\rho_\tau = c_{x_\tau}$.
Since the algorithm performs swaps, we have that
$B_t \ge c(S^t)$.
Also, let $\rho_{t^*+1} = \rho_{(u_{t^*+1},b_{t^*}(u_{t^*+1}))}$, and
$B_{t^*+1} = B_{t^*} + c_{u_{t^*+1}}$
We note that, by the definition of $u_{t^*+1}$, we have that $B'= B_{t^*+1} > B - c(Y) =B''$.
For $j=1, \ldots , B_{t^*+1}$, let $\gamma_j = {\rho_t}$ when
$j=B_{t-1}+1, \ldots, B_{t}$.
Using the above notation, we have that
\begin{equation}
\label{eq:g_value_sum_2t_star}
g(({\stts}\setminus \{ b_{t^*}(u_{t^*+1}) \}) \cup \{ u_{t^*+1} \}) =
\sum_{\tau=1}^{t^*+1} (B_\tau -B_{\tau -1}) \rho_{\tau} = \sum_{j=1}^{B'} \gamma_j,
\end{equation}
and for $t=1, \ldots , t^*$,
\begin{equation}
\label{eq:g_value_sum_2t}
g(S^{t}) = \sum_{\tau=1}^t  (B_\tau -B_{\tau -1}) \rho_{\tau} = \sum_{j=1}^{B_{t}} \gamma_j.
\end{equation}

Now, we note that
\begin{align*}
    \min_{s \in \left[1..B'\right]}  \left \{ \sum_{j=1}^{s-1} \gamma_j + \frac{B''}{2} \gamma_s  \right  \} & =
\min_{t \in \left[ 1..t^* \right]}
\left \{ \sum_{j=1}^{B_{t}} \gamma_j + \frac{B''}{2} \gamma_{B_{t+1}}   \right  \} \\
&=
\min_{t \in \left[ 1..t^* \right]}
\left \{ g({\stt}) +  \frac{B''}{2} {\rtto} \right  \}.
\end{align*}

Using Lemma~\ref{lemma:wolsey_ineq} and Inequality (\ref{eq:sopt_vs_stt_greedy}), we have
\begin{align*}
    \displaystyle{\frac{g(S^{t^*} \setminus \{ b_{t^*}(u_{t^*+1}) \}  \cup \{ u_{t^*+1} \})}{g({\sopt})} }
    & \geq \displaystyle{\frac{\sum_{j=1}^{B'} \gamma_j}{2 \left[  \min_{s\in \left[1..B'\right]} \sum_{j=1}^{s-1}
 \gamma_j + \frac{B''}{2} \gamma_s\right ]}} \\ \\
& \geq \frac{1}{2} \left(  1-(1- \frac{2}{B''})^{B'} \right ) \\ \\
& \geq \frac{1}{2} \left(  1 - e^{- \frac{2B'}{B''}} \right) \geq \frac{1}{2} (1 - e^{-2})
\end{align*}
Finally, using Lemma \ref{lemma:f_on_Y_swaps},
\begin{align*}
    f({S^{t^*}}) &= f(Y) +g(S^{t^*}) \\
                 &= f(Y) + g(({\stts} \setminus  \{ b_{t^*}(u_{t^*+1}) \}) \cup \{ u_{t^*+1} \}) \\
        &\phantom{= f(Y)}-\left(g(({\stts}\setminus  \{ b_{t^*}(u_{t^*+1}) \}) \cup \{ u_{t^*+1} \}) - g({\stts})\right) \\
                 &= f(Y) + g(({\stts} \setminus  \{ b_{t^*}(u_{t^*+1}) \}) \cup \{ u_{t^*+1} \}) \\
        &\phantom{= f(Y)}- \left(f(({\stts} \setminus  \{ b_{t^*}(u_{t^*+1}) \}) \cup \{ u_{t^*+1} \}) - f({\stts})\right) \\
                 &\geq  f(Y) + \frac{1}{2} (1 -e^{-2})g({\sopt}) - \frac{f(Y)}2 \\
                 &= \frac{f(Y)}2-\frac{f(Y)}{2} (1 -e^{-2})+\frac{1}{2} (1 -e^{-2}) f({\sopt}) \\
                 &\ge \frac{1}{2} (1 -e^{-2}) f({\sopt})
\end{align*}
\end{proof}

\mysubsection{Running Time Analysis}
\label{sec:runtime_single_matroid}

We note that the running time of Algorithm {\amatk}, as described above, may not be polynomial, since the number
of greedy swaps cannot be bounded polynomially.
Below, we show how the
algorithm can be implemented in polynomial-time with no harm to the approximation ratio.

\begin{theorem}
\label{thm:single_matroid_polytime}
Algorithm {\amatk} can be modified to run in ${\tilde{O}}(n^6)$ time.
\end{theorem}

\begin{proof}
To guarantee polynomial running time, we modify the greedy swaps as follows.
Fix some $\eps >0$ (to be set later). Whenever an improvement is
encountered, a swap is performed only if the value of $f(\cdot)$ increases by at least a
factor of 1+$\frac {\eps}{n^2}$ as a result of the swap.
(We note that this modification applies only to ``real'' swaps, and not
to swaps with the dummy element $\phi$.)
Let $\sopt= \{u_1, \ldots , u_p \}$, then by Lemma \ref{lemma:f_on_Y_swaps},
\begin{eqnarray}
f(\sopt) & =& f(\{ u_1 \}) + \sum_{i=2}^p (f(\{ u_1 , \ldots , u_i \}) -f(\{u_1 , \ldots , u_{i-1} \})) \nonumber  \\
& \leq & f(Y)+(|{\sopt} | -2 ) \frac{f(Y)}{2} \leq \frac{n}{2} f(Y)
\label{eqn:sopt_vs_Y}
\end{eqnarray}
Hence, the overall number of swaps is bounded by $\frac {n^2}\epsilon \log n$.
We show next that we can find a constant $\eps>0$ that would not degrade the
approximation ratio of the algorithm.
Consider the two cases in the proof of Theorem~\ref{thm:amatk_appx_ratio}.
In Case 1, after the modification Inequality~(\ref{eq:LS}) becomes
\begin{align}
\label{eq:LSm}
f(S^*)-f(S) & \leq \sum_{x \in\sopt\setminus S} \left( f(S \cup \{ x \}) -f(S)\right) \nonumber \\
            & \leq \sum_{x \in\sopt\setminus S}  \left( f(S \cup \{x \} \setminus  \{b(x) \}) - f(S \setminus \{ b(x) \})\right) \nonumber \\
            & \leq \sum_{x \in\sopt\setminus S}  \left(\left( 1+\frac{\epsilon}{n^2}\right) f(S) - f(S \setminus \{ b (x) \}) \right)  \nonumber \\
            & \leq \left(1+\frac\epsilon n\right) f(S).
\end{align}
Thus, we get $\frac 1{2+\epsilon n^{-1}}$-approximation. Since the approximation factor in Case 2 is
$\frac{1}{2} (1 -e^{-2})$, the overall approximation remains unchanged, as long as
$\frac 1{2+\epsilon n^{-1}} \ge \frac{1}{2} (1 -e^{-2})$, which implies
$\frac\epsilon n \le \frac 2{e^2 -1}$.

Now, consider Case 2 of Theorem~\ref{thm:amatk_appx_ratio}.
After the modification, Inequality~(\ref{eq:partition}) becomes

\begin{align}
\label{eq:partitionm}
g({\sopt}) - g({\stt}) &
\leq  \displaystyle { \sum_{x \in\sopt\setminus S} }
(g({\stt} \cup \{ x \}) - g({\stt}))  \nonumber\\
& \leq \displaystyle {\sum_{x \in C}}
\left(\left(1+\frac {\epsilon}{n^2}\right) g({\stt}) - g({\stt} \setminus \{ b_t(x) \})\right)  \nonumber\\
&+ \displaystyle {\sum_{x \in\sopt\setminus S\setminus C}}
\left(g({\stt}) -  g({\stt} \setminus \{ b_t(x) \})\right)+(g({\stt} \setminus \{ b_t(x) \} \cup \{ x \})  - g({\stt})).
\end{align}
Note that for any $x\in \sopt\setminus S$, if $b_t(x)=\phi$ then
$\rho_{(x,b_t(x))} \le \rho_{t+1}$, otherwise, i.e. $b_t(x)\neq\phi$, either
$g({\stt} \setminus \{ b_t(x) \} \cup \{ x \}) \le \left(1+\frac {\epsilon}{n^2}\right) g({\stt})$
or $\rho_{(x,b_t(x))} \le \rho_{t+1}$.
The set $C$ contains all the elements for which
the first inequality is satisfied.
Inequality~(\ref{eq:partitionm}) implies modified Inequality~(\ref{eq:rho}):
\begin{align}
\label{eq:rhom}
g({\sopt})
& \leq  (2+\epsilon n^{-1}) g({\stt}) + (B - c(Y)) \rho_{t+1} .
\end{align}

Proceeding to propagate this modification, we finally get
\begin{align*}
    f({S^{t^*}})
    &\geq  f(Y) + \frac{1}{2+\epsilon n^{-1}} \left(1 -\frac 1{e^{2}}\right)g({\sopt}) - \frac{f(Y)}{2} \\
    &\geq \frac{f(Y)}{(2+\epsilon n^{-1})e^2} +\frac{1}{2+\epsilon n^{-1}} \left(1 -\frac 1{e^{2}}\right) f({\sopt})
\end{align*}
It follows that the approximation ratio remains $\frac 12 (1-e^{-2})$,
as long as $(\epsilon n^{-1})(1-e^{-2})f({\sopt}) \le 2f(Y)e^{-2}$.
By (\ref{eqn:sopt_vs_Y}),
we have $nf(Y) \ge 2f({\sopt})$; thus, this holds for $\epsilon \le \frac 4{e^2-1}$.

We now turn to analyze the running time. We have $O(n^2)$ guesses of the set $Y$. For each
such guess, we have overall $\tilde{O}(n^2)$ successful swaps.
(Note that there are $O(n)$
swaps in which $\phi$ is involved.)
We have overall $\tilde{O}(n^2)$ swaps, each
requires $O(n^2)$ time.
We conclude that the overall running time is $\tilde{O}(n^6)$.
\end{proof}

We remark that a factor of $n$ can be saved in the running time at the expense of reducing
the approximation ratio by $O(\eps)$. Specifically, we can perform the swaps only if the multiplicative
improvement is 1+$\frac {\eps}{n}$.

\section{$k$ Matroid Constraints}
\label{sec:k_matroids}
In this section we extend our algorithm for a single matroid constraint to handle a single knapsack constraint and $k$ matroid constraints,  for any fixed $k> 1$; namely, we give an approximation
 algorithm for problem (\ref{eq:k_matroid_problem}).

\subsection{Algorithm}
\label{sec:alg_mati_k}
A key operation in our algorithm for  $k$ matroid constraints is $k$-swap,
in which we add to the solution set a single element and eliminate up to $k$ elements.
Denote the collection of subsets of $U$ of size at most $k$ by $[U]^{\le k}$;
note that $\emptyset\in [U]^{\le k}$.
Given a subset $U' \subseteq U$, $x \in U \setminus U'$ and  $\by\in [U']^{\le k}$,
$(x,{\by})$ is a {\em k-swap} if $(U' \setminus \by) \cup \{ x \}$ is independent, i.e.,
$(U' \setminus \by) \cup \{ x \} \in \bigcap_{j=1}^k  {\cF}_j $.

Let $L_k(U')$ denote the collection of $k$-swaps.
The marginal profit density of a $k$-swap $(x,{\by})$ is given by
$\rho_{(x,{\by})} = \frac{f((U'  \setminus  {\by})  \cup \{ x \}) - f(U')}{c_x}$.

\begin{algorithm}[!h]
\caption{{\amatik}$(U, B, {\cM}_1, \ldots ,  {\cM}_k)$ }
    \label{alg:submod_kmi}
\begin{algorithmic}[1]
    \State Guess $u_1 = \argmax_{u\in\sopt} f(\{u\})$ and $u_2 =   \argmax_{u\in\sopt\setminus{\{u_1\}}} f(\{u, u_1\}) - f(\{u_1\})$.
    \State Let $Y= \{u_1, u_2 \}$.
    \State Initialize $S = Y$ and ${\addls} = true$
    \While{${\addls}$}
        \Statex\Comment{greedy $k$-swaps}
        \State $\addls = false$
        \State Generate the collection of $k$-swaps $L=L_k(S)$
        \While{($not({\addls})$ and $L \neq \emptyset$)}
            \State Pop (i.e., pick and remove) from $L$ a swap $(x,\by)$ with a maximum value of $\rho_{(x,\by)}$
            \If{ ($\by\cap Y =\emptyset$ and $\rho_{(x,{\by})} >0$ and  $c_x  - c({\by}) + c(S) \leq B$)}
                 \State $S= (S  \setminus {\by})  \cup \{x\}$
                 \State $c(S)= c(S)  - c({\by}) + c_x$
                 \State  ${\addls}= true$
            \EndIf
        \EndWhile
    \EndWhile
    \State Return $S$
\end{algorithmic}
\end{algorithm}

\myparagraph{Overview}  Similar to {\amatk}, Algorithm {\amatik} modifies the solution set iteratively,
while increasing the objective function value and maintaining the knapsack and matroid constraints.
The main change
is in replacing swaps with $k$-swaps.
A pseudocode of {\amatik} is given in Algorithm~\ref{alg:submod_kmi}.

\subsection{Analysis}

We first analyze the performance ratio of {\amatik}.
\begin{theorem}
\label{thm:amatik_appx_ratio}
For any fixed $k \geq 1$,
{\amatik} is a $\frac{1-e^{-(k+1)}}{k+1}$-approximation algorithm for problem (\ref{eq:k_matroid_problem}).
\end{theorem}
We use in the proof the next lemmas.

\begin{lemma}
\label{lemma: partition_k_matroid}
Given two independent sets $S,T \in  \bigcap_{j=1}^k  {\cF}_j $,
there exists a mapping of the elements in $T \setminus S$ to $[S\setminus T]^{\le k}$
(namely, subsets of $S\setminus T$ of size at most $k$, including the empty set),
such that
each element $u \in S\setminus T$ appears in at most $k$ subsets.
\end{lemma}

\begin{proof}
For a matroid ${\cM}_j$, $1 \leq j \leq k$, since $S,T \in {\cF}_j$ are independent sets,
by Lemma~\ref{lemma: mapping_sopt_s} there exists
a mapping $b_j: T \setminus S \rightarrow (S \setminus T) \cup \{ \phi \}$,
such that $(S \setminus b_j(x)) \cup \{ x \} \in {\cF}_j$, for all
$x \in T \setminus S$.
We use the mappings $b_j$ to define the mapping
$b: T \setminus S \rightarrow [S\setminus T]^{\le k}$ as follows.
For each element $x \in T \setminus S$, define $b(x) = \bigcup_{j=1}^k \{b_j(x)\}\cap(S\setminus T)$.
The $k$-swaps are defined accordingly.
By  Lemma~\ref{lemma: mapping_sopt_s}, for each element $u\in S\setminus T$,
and for $1 \leq j \leq k$, $|b_j^{-1}(u)|\le 1$. It follows that
each element $u \in S\setminus T$ appears in at most $k$ mapped subsets.
\end{proof}

The proof of the next lemma is similar to the proof of Lemma~\ref{lemma:f_on_Y_swaps} (details omitted).

\begin{lemma}
\label{lemma:f_on_Y_k_swaps}
Given an element $u_\ell \in {\sopt}$,  $\ell \geq 3$, and
a subset $W \subseteq U \setminus \{ u_1, u_2, u_\ell\}$, for any $k$-swap  $(u_\ell , {\bar w})$, where
${\bar w} \in [W]^{\le k}$,
$
f((Y \cup W  \setminus {\bar w})  \cup \{ u_\ell \} ) - f(Y \cup W) \leq f(Y)/2.
$
\end{lemma}

\begin{proof}[Proof of Theorem \ref{thm:amatik_appx_ratio}]
Suppose we have guessed
$u_1 = \argmax_{u\in\sopt} f(\{u\})$ and
\\
$u_2 = \argmax_{u\in\sopt\setminus{\{u_1\}}} f(\{u, u_1\}) - f(\{u_1\})$ correctly.
Let $S\subseteq U$ denote the subset output by the algorithm.
By Lemma~\ref{lemma: partition_k_matroid}, there exists
a mapping $b: {\sopt}\setminus S \rightarrow [S\setminus\sopt]^{\le k}$,
such that $(S \setminus b(u)) \cup \{ u \} \in {\cF}$, for all $u \in {\sopt}\setminus S$.

We distinguish between two cases, based on the last iteration of the algorithm.
\par\noindent{\bf Case 1:}
For all elements $x\in \sopt\setminus S$
the $k$-swap $(x,b(x))$ was rejected
because it satisfied $\rho_{(x,b(x))} \le 0$
(and not because $c_x  - c(b(x)) + c(S) > B$.
That is, the swap $(x,b(x))$ did not violate the knapsack constraint.

By Case 1 assumption, it follows that no $k$-swap $(x,b(x))$, for
an element $x \in \sopt\setminus S$, could increase the value of $f(\cdot)$.
Hence, similar to the single matroid case, we have
\begin{align*}
f(\sopt)-f(S) & \leq \sum_{x \in\sopt\setminus S} \left( f(S \cup \{ x \}) -f(S)\right) \nonumber \\
            & \leq \sum_{x \in\sopt\setminus S}  \left( f(S \cup \{x \} \setminus  b(x) ) - f(S \setminus  b(x) )\right) \nonumber \\
            & \leq \sum_{x \in\sopt\setminus S}  \left( f(S) - f(S \setminus b (x)) \right) .
\end{align*}

In the next claim, we show that
$\sum_{x \in\sopt\setminus S}  \left( f(S) - f(S \setminus  b (x) ) \right) \le k\cdot f(S)$.
Thus, in this case, we have that {\amatik} yields a $\frac 1{k+1}$-approximation to the optimum.
\begin{claim}
\label{claim:kmat_local_search}
    Let $b:\sopt\setminus S \rightarrow [S\setminus\sopt]^{\le k}$ be the mapping defined in
    Lemma~\ref{lemma: partition_k_matroid}. Then
$\sum_{x \in\sopt\setminus S}  \left( f(S) - f(S \setminus b (x) ) \right) \le k\cdot f(S)$.
\end{claim}

\begin{proof}
Let $S = \{u_1,\ldots,u_{|S|}\}$, where the indices are assigned arbitrarily.
Fix an element $x \in \sopt\setminus S$, for which $b(x) \ne \emptyset$.
Let $b(x) = \{u_{i_1},\dots,u_{i_\ell}\}$, where $i_1 <\cdots< i_\ell$, and $\ell\le k$,
and let $b^j(x) = \{u_{i_j},\dots,u_{i_\ell}\}$, for $1 \le j \le \ell$,
and $b^{\ell+1}(x) =\emptyset$.
We have
\begin{align}
\label{eq:kmat_local_search}
    f(S) - f(S \setminus b (x))
     & =    \sum_{j=1}^\ell \left( f(S\setminus b^{j+1}(x)) -f(S\setminus b^j(x)) \right) \nonumber \\
     & \leq \sum_{j=1}^\ell \left( f(\{u_1,\ldots,u_{i_j} \}) - f(\{u_1,\ldots,u_{i_j -1} \}) \right)
\end{align}
The last inequality follows from submodularity, since for $1\le j \le \ell$,
\[
f(S\setminus b^{j+1}(x)) -f(S\setminus b^j(x)) \leq f(\{u_1,\ldots,u_{i_j} \}) - f(\{u_1,\ldots,u_{i_j -1}\}).
\]
As each element in $S\setminus\sopt$ appears in at most $k$ subsets $b(x)$, for
$x \in \sopt\setminus S$, summing Inequality~(\ref{eq:kmat_local_search}) over all elements
in $\sopt\setminus S$ we get
\[
    \sum_{x \in\sopt\setminus S}  \left( f(S) - f(S \setminus b (x) ) \right)
\le k\cdot\sum_{j=1}^{|S|}  \left( f(u_1,\ldots, u_j) - f(u_1,\ldots, u_{j-1})\right)
\le k\cdot f(S).
\]
\end{proof}

\par\noindent{\bf Case 2:}
At least one $k$-swap $(x,b(x))$, for $x\in \sopt\setminus S$, considered in the
last iteration,
violated the knapsack constraint; namely, satisfied
$c_x  - c(b(x)) + c(S) > B$.

As in the single matroid case,
denote by $S^\ell$ the subset of elements in the solution after iteration $\ell$
of the algorithm, for $\ell \geq 0$, where $S^0=Y$.
Let $b_\ell: \sopt\setminus S^\ell \rightarrow [S^\ell\setminus\sopt]^{\le k}$
be the mapping guaranteed by
Lemma~\ref{lemma: partition_k_matroid}.
Recall that
in all iterations $t$ except the last, the algorithm performs a swap
$(x_t,\by_t)$, where $\by_t=b(x_t)$.
Denote $\rho_{(x_t,\by_t)}$ by $\rho_t$.

Let $t^*+1$ be the first iteration in which
a swap $(u_{t^*+1},b_{t^*}(u_{t^*+1}))$, for $u_{t^*+1}\in \sopt\setminus S$,
was considered
and rejected, since it
violated the knapsack constraint.
By the choice of $t^*$, for all $t=0, \ldots , t^*-1$,
for every $x\in \sopt\setminus\stt$ such that $(x,b_t(x))$ is a swap, we must have
$\rho_{(x,b_t(x))} \le \rho_{t+1}$.

As in the single matroid case, define the function $g(S)=f(S)-f(Y)$.
Using submodularity, we can now write:
\begin{align}
\label{eq:partition}
g({\sopt}) - g({\stt})
& \leq \displaystyle {\sum_{x \in {{\sopt} \setminus {\stt}}}}
\left(g({\stt} \cup \{ x \}) - g({\stt})\right)  \nonumber\\
& \leq \displaystyle {\sum_{x \in {{\sopt} \setminus {\stt}}}}
\left(g(({\stt} \setminus \{ b_t(x) \})  \cup \{ x \}) - g({\stt} \setminus \{ b_t(x) \}) \right)  \nonumber\\
& \leq \displaystyle {\sum_{x \in {{\sopt} \setminus {\stt}}}}
\left(g({\stt}) -  g({\stt} \setminus \{ b_t(x) \}) \right)+
\left(g(({\stt} \setminus \{ b_t(x) \}) \cup \{ x \})  - g({\stt}) \right)
\end{align}

Following the proofs of claims~\ref{claim:bound_R1_plus_R2} and~\ref{claim:kmat_local_search}, we have
$\displaystyle {\sum_{x \in {\sopt} \setminus {\stt}}}
(g({\stt}) - g({\stt} \setminus b_t(x))) \leq k\cdot g({\stt})$.
Thus, for any $t=0, \ldots , t^*-1$,

\begin{align}
\label{eq:kmat_rho}
g({\sopt}) & \leq (k+1)g({\stt})
+ \sum_{x \in {\sopt} \setminus {\stt}}
\left(g(({\stt} \setminus  b_t(x)) \cup \{ x \})  - g({\stt}) \right) \nonumber\\
& \leq  (k+1)g({\stt}) + (B - c(Y)) \rho_{t+1}.
\end{align}

It follows that
\begin{equation}
\label{eq:kmat_sopt_vs_stt_greedy}
g({\sopt}) \leq (k+1)  \left ( g({\stt}) +\frac{B - c(Y)}{k+1} {\rtto} \right  )
\end{equation}
A derivation similar to the one for the single matroid case,
combined with Lemma~\ref{lemma:wolsey_ineq} and~(\ref{eq:kmat_sopt_vs_stt_greedy}), yields
\begin{align*}
    \displaystyle{\frac{g((S^{t^*} \setminus b_{t^*}(u_{t^*+1}))  \cup \{ u_{t^*+1} \})}{g({\sopt})} }
    & \geq \displaystyle{\frac{\sum_{j=1}^{B'} \gamma_j}{(k+1) \left[  \min_{s=1, \ldots , B'} \sum_{j=1}^{s-1}
 \gamma_j + \frac{B''}{(k+1)} \gamma_s\right ]}} \\
& \geq \frac{1}{k+1} \left(  1-(1- \frac{k+1}{B''})^{B'} \right ) \\ \\
& \geq \frac{1}{k+1} \left(  1 - e^{- \frac{(k+1)B'}{B''}} \right) \geq \frac{1}{k+1} (1 - e^{-(k+1)})
\end{align*}
Finally, using Lemma \ref{lemma:f_on_Y_k_swaps},

\begin{align*}
    f({S^{t^*}}) &= f(Y) + g(S^{t^*}) \\
                 &= f(Y) + g(({\stts} \setminus  b_{t^*}(u_{t^*+1})) \cup \{ u_{t^*+1} \}) \\
                 &\phantom{= f(Y)}
                 -\left(g(({\stts}\setminus b_{t^*}(u_{t^*+1}))\cup\{ u_{t^*+1} \}) - g({\stts})\right) \\
                 &= f(Y) + g(({\stts} \setminus   b_{t^*}(u_{t^*+1})) \cup \{ u_{t^*+1} \}) \\
                 &\phantom{= f(Y)}
                 -\left(f(({\stts}\setminus b_{t^*}(u_{t^*+1}))\cup\{ u_{t^*+1} \}) - f({\stts})\right) \\
                 &\geq  f(Y) + \frac{1}{k+1} (1 -e^{-(k+1)})g({\sopt}) - \frac{f(Y)}2
                 \ge \frac{1}{k+1} (1 -e^{-(k+1)}) f({\sopt})
\end{align*}

\end{proof}

\noindent
Using analogous techniques, as in the case of a single matroid constraint, we can modify the algorithm to obtain the following theorem.
\begin{theorem}
\label{thm:multi_matroid_polytime}
Algorithm {\amatik} can be modified to run in $\tilde{O}(n^{k+5})$ time.
\end{theorem}

\begin{proof}[Proof (sketch)]
Fix some $\eps >0$ (to be set later), and
perform a $k$-swap
only if the value of $f(\cdot)$ increases at least by a
factor of 1+$\frac {\eps}{n^2}$ as a result of the swap.
As before, the overall number of swaps is bounded by $\frac {n^2}\epsilon \log n$,
and it can be shown that for
$\eps \le \frac{2k+2}{e^2-1}$, the approximation ratio remains the same.

We now turn to analyze the running time. We have $O(n^2)$ guesses of the set $Y$. For each
such guess, we have overall $\tilde{O}(n^2)$ successful $k$-swaps
(Note that there are $O(n)$
swaps in which $\emptyset$ is involved.)
This implies, as before, an overall of $\tilde{O}(n^2)$
$k$-swaps.
Each $k$-swap requires $O(n^{k+1})$ time.
We conclude that the overall running time is $\tilde{O}(n^{k+5})$.
\end{proof}

\myparagraph{Acknowledgments}
We thank Chandra Chekuri and Moran Feldman for helpful comments on an earlier version of the paper.

\bibliography{bgap}

\end{document}